\documentclass[a4paper,10pt]{article}
\usepackage[utf8]{inputenc}
\usepackage{amssymb,amsthm,amsmath}

\usepackage[sort]{natbib}

\newtheorem{theorem}{Theorem}[section]
\newtheorem{defi}[theorem]{Definition}
\newtheorem{lemma}[theorem]{Lemma}
\newtheorem{pro}[theorem]{Proposition}
\newtheorem{cor}[theorem]{Corollary}
\newtheorem{re}[theorem]{Remark}

\begin{document}
\bibliographystyle{apalike}

\title{On optimal investment with processes of long or negative memory}
\author{Huy N. Chau\thanks{Alfr\'ed R\'enyi Institute of Mathematics, Hungarian Academy of Sciences, Budapest.
The authors were supported by the ``Lend\"ulet'' grant LP2015-6 of the Hungarian Academy of Sciences.
We thank Gordan \v{Z}itkovi\'c for helpful discussions related to this paper.}
\and Mikl\'os R\'asonyi\footnotemark[1]}

\maketitle

\begin{abstract} We consider the problem of utility maximization for investors with power utility functions.
Building on the earlier work \cite{larsen2014expansion}, 
we prove that the value of the problem
is a Fr\'echet-differentiable function of the drift of the price process, provided that this drift lies
in a suitable Banach space. 

We then study optimal investment problems with non-Markovian driving
processes. In such models there is no hope to get a formula for the achievable maximal utility.
Applying results of the first part of the paper we provide first order expansions for certain problems involving 
fractional Brownian motion either in the drift or in the volatility.
We also point out how asymptotic results can be derived for models with strong mean reversion.
\end{abstract}


\section{Introduction}

Portfolio optimization is one of the most important questions in mathematical finance, however, explicit solutions are obtained occasionally in some concrete and simple market models. As a tool for solving optimization problems, the dynamic programming principle shows its power in the Markovian realm. Nevertheless, the non-Markovian world exhibits significant difficulties which require different approaches. 

Recent literature on optimal investment has paid considerable attention to approximate solutions. 
More precisely, in many cases where no closed-form result can be found, series expansions could be provided for both the value function and the
parameters of the optimal strategy, see e.g. \cite{walter1,walter2,walter3,weber}. In other words, those difficult optimization problems
were considered as a perturbation of some well-understood benchmark in an appropriate manner. In the present article we wish to take
a similar, ``perturbation'' view on optimal investment in certain classes of non-Markovian market models.

Our starting point is the main result in \cite{larsen2014expansion}. In 
an incomplete market, where the stock price dynamics are given by a continuous semimartingale, these authors 
prove that a power investor's value function is G\^ateaux differentiable with respect to the market price of risk. 
We push this idea further by proving that the value function also has a Fr\'echet derivative on suitable domains. 
Since the space of market prices of risk is infinite dimensional, the 
discrepancy between these two concepts of differentiability is significant, see Remark \ref{arci} below. 

We then apply this new result to some particular market models exhibiting long memory or antipersistence (negative memory). 
Our examples will be driven by a fractional Brownian motion (fBm) whose memory features will be characterized by its Hurst parameter $H \in (0,1)$. 
It is well-known that fBms have long memory for $H>1/2$ and negative memory for $H<1/2$, see Chapter 3 of \cite{giraitis} for
a detailed discussion of these concepts, especially Definitions 3.1.2 and 3.1.3.

One of the principal messages of our paper is that the value function of a power utility maximizer is differentiable with respect to 
the Hurst parameter, with an explicitly given derivative.
This helps to deepen our knowledge on the effect of memory in optimization. Some further applications will also be mentioned.

The paper is organized as follows. In Section \ref{sec_setting}, we introduce the market model and recall the pertinent 
terminology. In Section \ref{sec_frechet}, we extend the result of \cite{larsen2014expansion} by proving that the value function is also Fr\'echet differentiable. 
Section \ref{sec_fBm} introduces some financial markets driven by fractional Brownian motion.  
Section \ref{sec_further} treats models with strongly mean-reverting drifts. 
Section \ref{conci} comments on possible future work.
Section \ref{sec_appen} provides some necessary computations with fBms. 

\section{The model}\label{sec_setting}

Let $(\Omega, \mathcal{F},\mathbb{P})$ be a probability space, and $\mathbb{F} := (\mathcal{F}_t)_{t \in [0,T]}$ be a filtration satisfying the usual conditions. 
Let $\mathbb{L}^0$ denote the set of (a.s. equivalence classes of) scalar random variables, endowed with the topology of convergence in probability.
The symbol $\mathbb{L}^p$ refers to the usual Banach space of $p$-integrable random variables.

For a continuous $\mathbb{F}$-local martingale $M$, we denote
$$
\mathcal{P}^2_M := \left\lbrace  (\pi_t)_{t\in [0,T]}:\ \pi\text{ is }\mathbb{F}-\text{progressively measurable},\ 
\int\limits_0^T {\pi^2_td \left\langle M \right\rangle _t} < \infty, \mathbb{P}-a.s. \right\rbrace,
$$
where $\left\langle M \right\rangle$ denotes the quadratic variation of the local martingale $M$.
For each $\lambda \in \mathcal{P}^2_M$, we define a continuous semimartingale $S^{\lambda} = \mathcal{E}(R^{\lambda})$ where 
$$R^{\lambda}_t= \int\limits_0^t {\lambda_s d\left\langle M\right\rangle}_s + M_t,$$
and $\mathcal{E}$ denotes the stochastic exponential.
The process $R^{\lambda}$ is interpreted as the return of the process $S^{\lambda}$ and the process $\lambda$ is called the market price of risk in the literature.
The financial market consists of the risky asset $S^{\lambda}$ and a zero-interest bond. Given $x>0$, we denote by $\mathcal{X}^{\lambda}(x)$ the set of all 
nonnegative wealth processes starting from initial capital $x$, i.e. 
$$\mathcal{X}^{\lambda}(x):=\left\lbrace  x \mathcal{E}\left( \int\limits_0^{\cdot} {\pi_t dR^{\lambda}_t}\right): \pi \in \mathcal{P}^2_M \right\rbrace.$$
\textbf{The optimization problem.} The investor's preferences are modeled by a power utility function
$$U(x):=\frac{x^p}{p},\qquad x >0,$$
where $p<0$ is the risk-aversion parameter. For $x>0$, we consider the optimal investment problem
\begin{equation}\label{primal}
u^{\lambda}(x):=\max_{X \in \mathcal{X}^{\lambda}(x)} \mathbb{E}U \left( X_T \right).
\end{equation}
\textbf{The dual optimization problem.} For $\lambda \in \mathcal{P}$, the stochastic exponential $Z^{\lambda}$ given by
$$Z^{\lambda}_t := \mathcal{E}\left( - \lambda \cdot M \right)_t, \qquad t \in [0,T],$$
is a strictly positive local martingale. This process figures prominently in the study of the dual domain, which  is described now. 
For $y>0$, we define 
\begin{align*}
\mathcal{Y}^{\lambda}(y) := \bigl\{ Y >0 :\;& Y \text{ is a c\`adl\`ag supermartingale such that } Y_0 =y \text{ and } \\
& YX \text{ is a supermartingale for every } X \in \mathcal{X}^{\lambda}(1) \bigr\}.
\end{align*}
The conjugate function $V: (0,\infty) \to \mathbb{R}$ is defined by
\begin{equation}
V(y):=\sup_{x>0}(U(x)-xy)= \frac{y^{-q}}{q}, \qquad  \text{where }q:=\frac{p}{1-p} \in (-1,0).
\end{equation}
The dual value function is defined as follows
\begin{equation}\label{dual}
v^{\lambda}(y):= \inf_{Y \in \mathcal{Y}^{\lambda}(y)} \mathbb{E}[V(Y_T)], \qquad y > 0.
\end{equation}
It is shown in Corollary 3.3 of \cite{larsen2007stability} that if $v^{\lambda}(y)$ is finite then the minimizer  $\widehat{Y}^{\lambda}$ for the dual problem always has the form 
$$\widehat{Y}^{\lambda} = Z^{\lambda}\widehat{H}^{\lambda}, \qquad \text{for some } \widehat{H} \in \mathcal{H},$$
where $\mathcal{H}$ denotes the set of all positive local martingales $H$ with $H_0= 1$ that are strongly orthogonal to $M$.
\begin{re}[No arbitrage conditions] 
{\rm The requirement that $\lambda \in \mathcal{P}^2_M$ is not enough for the validity of the NFLVR condition, see \cite{ds},
since a candidate for a local martingale density can easily be a strict local martingale.  However, this is the minimal condition needed to ensure 
that the optimization problem is well-posed. Indeed, the process $Z^{\lambda}$ plays the role of a supermartingale deflator, 
in the terminology of \cite{karatzas2007numeraire}, and its presence implies that the condition NUPBR of that paper holds true. For recent duality results under the condition NUPBR, we refer to \cite{chau2015optimal}.} 
\end{re}

\section{Fr\'echet differentiability}\label{sec_frechet}


\cite{larsen2014expansion} prove that the value function $u^{\lambda}$ is G\^ateaux differentiable with respect to $\lambda$ on a certain, suitable domain. 
We will show that even a Fr\'echet derivative can be found in a smaller domain, which requires stronger regularity of the market price of risk.

Let $\alpha >1$ be a real number such that $(-q )\alpha < 1$ and let $\beta > 1$ be the H\"older conjugate of $\alpha,$ i.e. $\frac{1}{\alpha} + \frac{1}{\beta} = 1$. In short, let $\beta>1$ be such that $\frac{-q\beta}{\beta-1}<1$. For such $\beta$
we define 
\begin{equation}\label{norm_D}
\left\| \lambda \right\|_{\beta}:=  \mathbb{E}\left[ \left( \int\limits_0^T {(\lambda_t)^2d\left\langle M \right\rangle _t}\right) ^{\beta}  \right]^{\frac{1}{2\beta}}, 
\end{equation}
and set $\mathcal{D}_{\beta} := \left\lbrace  \lambda \in \mathcal{P}^2_M: \left\| \lambda \right\|_{\beta} < +\infty\right\rbrace$.
It is easily checked that $\mathcal{D}_{\beta}$ is a Banach space with the given norm.

Following Definition 2.7 of \cite{larsen2007stability}, a metrizable topology $\tau$ on some subset 
$\mathcal{A}\subset\mathcal{P}^2_M$ is said to be appropriate if the mapping $\lambda \mapsto Z^{\lambda}_T$ from $\mathcal{A}$ into $\mathbb{L}^0_+$ is continuous, when $\mathcal{A}$ is endowed with $\tau$, and $\mathbb{L}^0$ with the topology of convergence in probability.
\begin{lemma}\label{nona}
	The norm given in (\ref{norm_D}) induces an appropriate topology on $\mathcal{D}_{\beta}$.
\end{lemma}
\begin{proof}
Suppose that $\lambda^n \to \lambda$ in $\mathcal{D}_{\beta}$. H\"older's inequality shows that
$$ \mathbb{E}\left[ \int_0^T {|\lambda^n_t - \lambda_t|^2} d\left\langle M \right\rangle _t \right] \le \mathbb{E}\left[ \left(  \int_0^T {|\lambda^n_t - \lambda_t|^2} d\left\langle M \right\rangle _t \right)^{\beta} \right]^{1/\beta} \to 0.$$
Since the norm $\mathbb{E}\left[ \int_0^T {\lambda_t^2} d\left\langle M \right\rangle _t \right]^{\frac{1}{2}}$ induces an appropriate topology by Proposition A.1 of 
\cite{larsen2007stability}, it holds that $Z^{\lambda^n}_T \to Z^{\lambda}_T$ in probability. Therefore, the norm (\ref{norm_D}) also induces an appropriate topology.
\end{proof}
Now, we recall a result of \cite{larsen2007stability}.

\begin{theorem}\label{continuous}
For any $\lambda \in \mathcal{D}_{\beta}$, the function $u^{\lambda}:(0, \infty) \to \mathbb{R}$ is finite-valued. The mapping $(\lambda, x) \to u^{\lambda}(x)$ is jointly continuous.
\end{theorem}
\begin{proof} This follows from Theorem 2.12 of \cite{larsen2007stability} noting that
the set $\mathcal{D}_{\beta}$ is trivially $V$-relatively compact in the present case (see Definition 2.10 of 
\cite{larsen2007stability}) and the topology induced by (\ref{norm_D}) is appropriate,
see  Lemma \ref{nona} above. 
\end{proof}

For simplicity, we also write $u^{\lambda}:=u^{\lambda}(1)$ and $v^{\lambda}:=v^{\lambda}(1)$.
Furthermore, \cite{larsen2014expansion} prove that the value function is not only continuous but also G\^ateaux differentiable with respect to the market price of risk.
We recall their Theorem 3.1.

\begin{theorem}[The G\^ateaux derivative]\label{gateaux} Let $\lambda, \lambda' \in \mathcal{P}^2_M$ such that
\begin{equation}\label{mak}
\int\limits_0^T{ (\lambda'_t)^2}d\left\langle M \right\rangle_t \in \mathbb{L}^{1-p} \qquad \text{and} \qquad \int\limits_0^T {\lambda'_t}dR^{\lambda}_t  \in \cup_{s > 1-p} \mathbb{L}^s.
\end{equation}
Then we have 
\begin{align*} 
\left. \frac{d}{d\varepsilon}u^{^{\lambda + \varepsilon \lambda'}}\right|_{\varepsilon = 0}:= \lim_{\varepsilon \to 0} \frac{1}{\varepsilon}\left( u^{\lambda + \varepsilon \lambda'} - u^{\lambda} \right) = pu^{\lambda} \mathbb{E}\left[ \frac{d\tilde{\mathbb{P}}^{\lambda}}{d\mathbb{P}}\int\limits_0^T {\lambda_t'dR^{\lambda}_t} \right], \\
\left. \frac{d}{d\varepsilon}v^{\lambda + \varepsilon \lambda'}\right|_{\varepsilon = 0}:= \lim_{\varepsilon \to 0} \frac{1}{\varepsilon}\left( v^{\lambda + \varepsilon \lambda'} - v^{\lambda} \right) = qv^{\lambda} \mathbb{E}\left[ \frac{d\tilde{\mathbb{P}}^{\lambda}}{d\mathbb{P}}\int\limits_0^T {\lambda_t'dR^{\lambda}_t} \right],
\end{align*}
where the probability measure $\tilde{\mathbb{P}}^{\lambda}$ is defined by $\frac{d\tilde{\mathbb{P}}^{\lambda}}{d\mathbb{P}} = \frac{1}{u^{\lambda}}U(\widehat{X}^{\lambda}_T) = \frac{1}{v^{\lambda}} V(\widehat{Y}^{\lambda}_T),$ and $\widehat{X}^{\lambda}_T, \widehat{Y}^{\lambda}_T$ are the solutions of (\ref{primal}) and (\ref{dual}), respectively.
\end{theorem}
By the choice of $\beta$, we see that $1-p < \beta$, therefore any $\lambda' \in \mathcal{D}_{\beta}$ satisfies also the integrability conditions 
\eqref{mak} in Theorem \ref{gateaux}. Now, we push the idea of \cite{larsen2014expansion} further by proving that the value function is also Fr\'echet differentiable.
\begin{theorem}[The Fr\'echet derivative]\label{frechet} Let $\frac{-q\beta}{\beta-1}<1$ hold. Then
the mapping $\lambda \mapsto u^{\lambda}(x)$ from $\mathcal{D}_{\beta}$ to $\mathbb{R}$ is Fr\'echet differentiable.
\end{theorem}
\begin{proof}
Using \cite{larsen2014expansion}, the G\^ateaux derivative of $u^{\lambda}$ at $\lambda$ in the direction $\lambda'$ 
(see Theorem \ref{gateaux}) can be rewritten as
\begin{equation}\label{gateux2}
p\frac{u^{\lambda}}{v^{\lambda}} \mathbb{E}\left[ V(\widehat{Y}^{\lambda}_T)\int\limits_0^T {\lambda_t'dR^{\lambda}_t } \right].
\end{equation} 
First, we prove that the G\^ateaux derivative, as a linear operator acting on $\lambda'\in\mathcal{D}_{\beta}$ 
is continuous, at every point $\lambda\in\mathcal{D}_{\beta}$. Indeed, the linearity of (\ref{gateux2}) with respect to $\lambda'$ is obvious and thus we only need to consider its continuity on $\mathcal{D}_{\beta}$. Let $\lambda'' \to \lambda'$ in 
$\mathcal{D}_{\beta}$. Using the H\"older inequality and the Burkholder-Davis-Gundy inequality, we estimate
\begin{align*}
&\mathbb{E}\left[ V(\widehat{Y}^{\lambda}_T)\int\limits_0^T {|\lambda_t'- \lambda_t''|dR^{\lambda}_t } \right] \le \mathbb{E}\left[ V(\widehat{Y}^{\lambda}_T)^{\alpha} \right]^{\frac{1}{\alpha}} \mathbb{E}\left[\left|  \int\limits_0^T {|\lambda_t'- \lambda_t''|\lambda_t d\left\langle M \right \rangle _t }\right| ^{2\beta} \right]^{\frac{1}{2\beta}}  \\
&+ \mathbb{E}\left[ V(\widehat{Y}^{\lambda}_T)^{\alpha} \right]^{\frac{1}{\alpha}} \mathbb{E}\left[\left|  \int\limits_0^T {|\lambda_t'- \lambda_t''|dM_t }\right| ^{2\beta} \right]^{\frac{1}{2\beta}}\\
&\le \mathbb{E}\left[ V(\widehat{Y}^{\lambda}_T)^{\alpha} \right]^{\frac{1}{\alpha}} \mathbb{E}\left[\left|  \int\limits_0^T {|\lambda_t'- \lambda_t''|^2 d\left\langle M \right \rangle _t }\right| ^{\beta} \right]^{\frac{1}{2\beta}} \mathbb{E}\left[\left|  \int\limits_0^T {|\lambda_t|^2 d\left\langle M \right \rangle _t }\right| ^{\beta} \right]^{\frac{1}{2\beta}}  \\
&+ \mathbb{E}\left[ V(\widehat{Y}^{\lambda}_T)^{\alpha} \right]^{\frac{1}{\alpha}} \mathbb{E}\left[\left|  \int\limits_0^T {|\lambda_t'- \lambda_t''|^2d\left\langle M\right\rangle _t }\right| ^{\beta} \right]^{\frac{1}{2\beta}},
\end{align*}
which tends to zero and the first claim is proved. 

Now it suffices to prove that the G\^ateaux derivative is a continuous function of $\lambda$, see Proposition 3.2.15 of \cite{drabek2013methods}, 
i.e. we need to prove that the quantity in (\ref{gateux2}) is continuous on $\mathcal{D}_{\beta}$ (in the operator norm topology). 

Letting $\tilde{\lambda} \to 0$ in $\mathcal{D}_{\beta}$, we compute
\begin{align}\label{gateux_diff}
\frac{u^{\lambda + \tilde{\lambda}}}{v^{\lambda + \tilde{\lambda}}} &\mathbb{E}\left[ V(\widehat{Y}^{\lambda + \tilde{\lambda}}_T) \int\limits_0^T {\lambda_t'dR^{\lambda + \tilde{\lambda}}_t } \right] - \frac{u^{\lambda}}{v^{\lambda}} \mathbb{E}\left[ V(\widehat{Y}^{\lambda}_T)\int\limits_0^T {\lambda_t'dR^{\lambda}_t } \right] \nonumber\\
= &\mathbb{E}\left[ \frac{u^{\lambda + \tilde{\lambda}}}{v^{\lambda + \tilde{\lambda}}} V(\widehat{Y}^{\lambda + \tilde{\lambda}}_T)\int\limits_0^T {\lambda_t'dM_t} - \frac{u^{\lambda}}{v^{\lambda}} V(\widehat{Y}^{\lambda}_T)\int\limits_0^T {\lambda_t'dM_t } \right] \nonumber \\
+&\mathbb{E}\left[ \frac{u^{\lambda + \tilde{\lambda}}}{v^{\lambda + \tilde{\lambda}}} V(\widehat{Y}^{\lambda + \tilde{\lambda}}_T)\int\limits_0^T {\lambda_t' (\lambda_t + \tilde{\lambda}_t) d\left\langle M \right\rangle_t } - \frac{u^{\lambda}}{v^{\lambda}} V(\widehat{Y}^{\lambda}_T)\int\limits_0^T {\lambda_t'\lambda_t  d\left\langle M \right\rangle_t }\right]. 
\end{align}
Using H\"older's inequality, the first term in the RHS of (\ref{gateux_diff}) is bounded by
\begin{align*}
&\mathbb{E}\left[ \left| \left( \frac{u^{\lambda + \tilde{\lambda}}}{v^{\lambda + \tilde{\lambda}}} V(\widehat{Y}^{\lambda + \tilde{\lambda}}_T) - \frac{u^{\lambda}}{v^{\lambda}} V(\widehat{Y}^{\lambda}_T)  \right)  \int\limits_0^T {\lambda_t'dM_t}\right|  \right] \\
\le &\mathbb{E}\left[ \left| \frac{u^{\lambda + \tilde{\lambda}}}{v^{\lambda + \tilde{\lambda}}} V(\widehat{Y}^{\lambda + \tilde{\lambda}}_T) - \frac{u^{\lambda}}{v^{\lambda}} V(\widehat{Y}^{\lambda}_T)  \right|^{\alpha} \right]^{\frac{1}{\alpha}} \mathbb{E}\left[ \left| \int\limits_0^T {\lambda_t'dM_t}\right|^{2\beta}  \right]^{\frac{1}{2\beta}} \\
\le &C \mathbb{E}\left[ \left| \frac{u^{\lambda + \tilde{\lambda}}}{v^{\lambda + \tilde{\lambda}}} V(\widehat{Y}^{\lambda + \tilde{\lambda}}_T) - \frac{u^{\lambda}}{v^{\lambda}} V(\widehat{Y}^{\lambda}_T)  \right|^{\alpha} \right]^{\frac{1}{\alpha}} \mathbb{E}\left[ \left( \int\limits_0^T {(\lambda_t')^2d\left\langle M \right\rangle _t}\right) ^{\beta}  \right]^{\frac{1}{2\beta}},
\end{align*}
where the last bound is due to the Burkholder-Davis-Gundy inequality.
By Theorem \ref{continuous}, we have that $u^{\lambda + \tilde{\lambda}} \to u^{\lambda}$ as $\tilde{\lambda} \to 0$ in 
$\mathcal{D}_{\beta}$. In addition, $V(\widehat{Y}^{\lambda + \tilde{\lambda}}_T) \to V(\widehat{Y}^{\lambda}_T)$ in probability, see Lemma 3.5 of \cite{larsen2007stability}. This implies 
$$\frac{u^{\lambda + \tilde{\lambda}}}{v^{\lambda + \tilde{\lambda}}} V(\widehat{Y}^{\lambda + \tilde{\lambda}}_T) - \frac{u^{\lambda}}{v^{\lambda}} V(\widehat{Y}^{\lambda}_T) \to 0, \qquad \text{in probability as } \tilde{\lambda} \to 0.$$
Furthermore, the inequality $|x+y|^{\alpha} \le 2^{\alpha-1}(|x|^{\alpha} + |y^{\alpha}|)$ (with $\alpha > 1$) yields
\begin{equation}\label{uniform}
\left| \frac{u^{\lambda + \tilde{\lambda}}}{v^{\lambda + \tilde{\lambda}}} V(\widehat{Y}^{\lambda + \tilde{\lambda}}_T) - \frac{u^{\lambda}}{v^{\lambda}} V(\widehat{Y}^{\lambda}_T) \right|^{\alpha} \le 
2^{\alpha - 1}\left( \left(  \frac{u^{\lambda + \tilde{\lambda}}}{v^{\lambda + \tilde{\lambda}}} \right)^{\alpha}  (\widehat{Y}^{\lambda + \tilde{\lambda}}_T)^{-q\alpha} + \left( \frac{u^{\lambda}}{v^{\lambda}}\right)^{\alpha}  (\widehat{Y}^{\lambda}_T)^{-q\alpha}\right).
\end{equation}
It is noticed that $(-q) \alpha < 1,$ and $\sup_{\tilde{\lambda}}\mathbb{E}[\widehat{Y}^{\lambda + \tilde{\lambda}}_T] \le 1$, the theorem of de la Vall\'ee-Poussin implies that the family 
$(\widehat{Y}^{\lambda + \tilde{\lambda}}_T)^{-q \alpha}$ is uniformly integrable. Together with (\ref{uniform}), these imply
$$\mathbb{E}\left[ \left| \frac{u^{\lambda + \tilde{\lambda}}}{v^{\lambda + \tilde{\lambda}}} V(\widehat{Y}^{\lambda + \tilde{\lambda}}_T) - \frac{u^{\lambda}}{v^{\lambda}} V(\widehat{Y}^{\lambda}_T)  \right|^{\alpha} \right]^{\frac{1}{\alpha}} \to 0.$$
The second term on the RHS of (\ref{gateux_diff}) is estimated as follows,
\begin{align*}
&\mathbb{E}\left[\left|  \frac{u^{\lambda + \tilde{\lambda}}}{v^{\lambda + \tilde{\lambda}}} V(\widehat{Y}^{\lambda + \tilde{\lambda}}_T)\int\limits_0^T {\lambda_t' (\lambda_t + \tilde{\lambda}_t) d\left\langle M \right\rangle_t } - \frac{u^{\lambda}}{v^{\lambda}} V(\widehat{Y}^{\lambda}_T)\int\limits_0^T {\lambda_t'\lambda_t  d\left\langle M \right\rangle_t }\right| \right]\\
\le &\mathbb{E}\left[\left|  \frac{u^{\lambda + \tilde{\lambda}}}{v^{\lambda + \tilde{\lambda}}} V(\widehat{Y}^{\lambda + \tilde{\lambda}}_T)\int\limits_0^T {\lambda_t' (\lambda_t + \tilde{\lambda}_t) d\left\langle M \right\rangle_t } - \frac{u^{\lambda + \tilde{\lambda}}}{v^{\lambda + \tilde{\lambda}}} V(\widehat{Y}^{\lambda+ \tilde{\lambda}}_T)\int\limits_0^T {\lambda_t'\lambda_t  d\left\langle M \right\rangle_t }\right| \right] \\
+ &\mathbb{E}\left[\left|  \frac{u^{\lambda + \tilde{\lambda}}}{v^{\lambda + \tilde{\lambda}}} V(\widehat{Y}^{\lambda+ \tilde{\lambda}}_T)\int\limits_0^T {\lambda_t'\lambda_t  d\left\langle M \right\rangle_t } - \frac{u^{\lambda}}{v^{\lambda}} V(\widehat{Y}^{\lambda}_T)\int\limits_0^T {\lambda_t'\lambda_t  d\left\langle M \right\rangle_t }\right| \right].
\end{align*}
By using the generalized H\"older inequality in the same way as above, we obtain the following estimations
\begin{align*}
&\mathbb{E}\left[\left|  \frac{u^{\lambda + \tilde{\lambda}}}{v^{\lambda + \tilde{\lambda}}} V(\widehat{Y}^{\lambda+ \tilde{\lambda}}_T) \int\limits_0^T {\lambda_t'\lambda_t  d\left\langle M \right\rangle_t } - \frac{u^{\lambda}}{v^{\lambda}} V(\widehat{Y}^{\lambda}_T)\int\limits_0^T {\lambda_t'\lambda_t  d\left\langle M \right\rangle_t }\right| \right] \le \\
&\mathbb{E}\left[\left|  \frac{u^{\lambda + \tilde{\lambda}}}{v^{\lambda + \tilde{\lambda}}} V(\widehat{Y}^{\lambda+ \tilde{\lambda}}_T) - \frac{u^{\lambda}}{v^{\lambda}} V(\widehat{Y}^{\lambda}_T)\right|^{\alpha} \right]^{\frac{1}{\alpha}} \mathbb{E}\left[ \left( \int\limits_0^T {\lambda_t^2d\left\langle M \right\rangle _t}\right) ^{\beta}  \right]^{\frac{1}{2\beta}} \mathbb{E}\left[ \left( \int\limits_0^T {(\lambda_t')^2d\left\langle M \right\rangle _t}\right) ^{\beta}  \right]^{\frac{1}{2\beta}} 
\end{align*}
and 
\begin{align*}
&\mathbb{E}\left[\left|  \frac{u^{\lambda + \tilde{\lambda}}}{v^{\lambda + \tilde{\lambda}}} V(\widehat{Y}^{\lambda + \tilde{\lambda}}_T)\int\limits_0^T {\lambda_t' (\lambda_t + \tilde{\lambda}_t) d\left\langle M \right\rangle_t } - \frac{u^{\lambda + \tilde{\lambda}}}{v^{\lambda + \tilde{\lambda}}} V(\widehat{Y}^{\lambda+ \tilde{\lambda}}_T)\int\limits_0^T {\lambda_t'\lambda_t  d\left\langle M \right\rangle_t }\right| \right] \le \\
&\mathbb{E}\left[\left|  \frac{u^{\lambda + \tilde{\lambda}}}{v^{\lambda + \tilde{\lambda}}} V(\widehat{Y}^{\lambda+ \tilde{\lambda}}_T)\right|^{\alpha} \right]^{\frac{1}{\alpha}} \mathbb{E}\left[ \left( \int\limits_0^T {\tilde{\lambda}_t^2d\left\langle M \right\rangle _t}\right) ^{\beta}  \right]^{\frac{1}{2\beta}} \mathbb{E}\left[ \left( \int\limits_0^T {(\lambda_t')^2d\left\langle M \right\rangle _t}\right) ^{\beta}  \right]^{\frac{1}{2\beta}}.
\end{align*}
From these estimations and (\ref{gateux2}), we conclude that the G\^ateaux derivative is continuous and hence the proof is complete.
\end{proof}

\begin{re}\label{arci}{\rm We wish to point out that Theorem \ref{frechet} is not a mere technical improvement with respect
to previous results (differentiability in a stronger sense is established) but it opens the door to 
essentially new developments that were hitherto inaccessible.

We briefly explain a general scheme here that can be implemented in various models. 
Sections \ref{sec_fBm} and \ref{sec_further} below carry out this scheme in concrete models, but the main idea
seems to be applicable in great generality.

Let $\Theta\subset\mathbb{R}^m$ be an open set and consider a parametric class of market models
of the type as in Section \ref{sec_setting}, i.e. let $\upsilon:\Theta\to\mathcal{D}_{\beta}$ be a mapping.
The market model corresponding to parameter $\theta\in\Theta$ will have drift $\upsilon(\theta)$.

Assume that $\upsilon$ is differentiable. Then the composition of $\upsilon$ with the mapping
$\lambda\to u^{\lambda}$, $\lambda\in \mathcal{D}_{\beta}$ is also differentiable by Theorem
\ref{frechet} (note that for this conclusion to hold, G\^ateaux differentiability does not suffice).
It follows that we can find a first order expansion of the mapping $\theta\to u^{\upsilon(\theta)}$ in
$\theta$, i.e. the sensitivity of the optimal expected utility with respect to the model parameter can be studied. 
We stress that it is Theorem \ref{frechet} that makes such an approach
feasible: G\^ateaux differentiability allows the study of linear perturbations for the drift while Fr\'echet 
differentiability permits non-linear perturbations as well.}
\end{re}

\begin{re}\label{mascot}
	{\rm Fix $\lambda$. Denoting by $\pi^{\lambda}$ the optimal strategy
		corresponding to drift $\lambda$, one expects that $\pi^{\lambda}$
		is nearly optimal for the problem with drift $\lambda'$ as well, provided
		that the latter is ``close enough'' to $\lambda$.
		
		When we fix a direction $\lambda'$ and we are only interested in approximating
		$\pi^{\lambda+\varepsilon\lambda'}$ then this issue has been thoroughly
		studied in Subsection 3.2 of \cite{larsen2014expansion}. They obtain,
		under suitable technical conditions, that $\pi^{\lambda}$ is $O(\varepsilon^2)$-optimal
		for the problem with drift $\lambda+\varepsilon \lambda'$. Moreover,
		in a Brownian setting (see \cite{larsen2014expansion}
		for details), a correction term can be found that improves on this:
		\begin{equation}\label{lal}
		{\pi}^{\lambda}+\varepsilon \frac{\lambda'+p\gamma^B}{1-p}
		\end{equation}
		is an $O(\varepsilon^3)$-optimal control for the drift $\lambda+\varepsilon\lambda'$.
		Here $\gamma^B$ is a process coming from the martingale representation
		theorem applied to $\int_0^T {\pi}^{\lambda}_t\lambda_t' d\langle M\rangle_t$ under
		a certain probability $\tilde{\mathbb{P}}^{\lambda}$ (see Theorem \ref{gateaux}) 
		that is constructed from the solution of the
		utility maximization problem at $\lambda$. 
		
		As seen from the above, the correction term is not simple because $\gamma^B$
		can be found only in very specific cases. The sufficient conditions for these results to hold are 
		difficult to check in concrete models due to the presence of $\tilde{\mathbb{P}}^{\lambda}$ on which one has little
		grasp. 
		
		Nonetheless we quickly sketch a (rather trivial and rather stringent) sufficient condition that ensures
		$O(\varepsilon)$ optimality of $\pi^{\lambda}$ for $\lambda'$ with $\Vert\lambda'-\lambda\Vert_{\beta}<\varepsilon$,
		\emph{irrespective} of the direction of $\lambda'-\lambda$. 
		For some more general comments on related problems, see Section \ref{conci}. Let $\lambda,\lambda'
		\in\mathcal{D}_{\beta}$. 
		
		One can write
		\begin{eqnarray*}
			\left|U\left(\mathcal{E}\left(\int_0^T \pi^{\lambda}_t\, dR^{\lambda}_t\right)\right)-
			U\left(\mathcal{E}\left(\int_0^T \pi^{\lambda}_t\, dR^{\lambda'}_t\right)\right)\right| &=&\\
			\frac{1}{|p|}\exp\left\{p\int_0^T \pi_t^{\lambda}\,dM_t+p\int_0^T \pi_t^{\lambda} \lambda_t\, d\langle M\rangle_t
			-\frac{p}{2}\int_0^T (\pi_t^{\lambda})^2\, d\langle M\rangle_t\right\} &\times &\\
			\left|\exp\left\{p\int_0^T \pi_t^{\lambda} (\lambda_t'-\lambda_t)\, d\langle M\rangle_t
			\right\}-1\right|
		\end{eqnarray*}
		so, using $|e^x-1|\leq |x| \sup_{|\varepsilon| \le 1} e^{\varepsilon |x|}$ and the Cauchy inequality, 
		\begin{eqnarray*}
			\mathbb{E}\left|U\left(\mathcal{E}\left(\int_0^T \pi^{\lambda}_t\, dR^{\lambda}_t\right)\right)-
			U\left(\mathcal{E}\left(\int_0^T \pi^{\lambda}_t\, dR^{\lambda'}_t\right)\right)\right| \leq \\
			\frac{1}{|p|}\mathbb{E}^{1/2} \exp\left\{2p\int_0^T \pi_t^{\lambda}\,dM_t+2p\int_0^T \pi_t^{\lambda} 
		\lambda_t\, d\langle M\rangle_t
			-p\int_0^T (\pi_t^{\lambda})^2\, d\langle M\rangle_t\right\}   \\ 
			  \times \\
			\vert p\vert\mathbb{E}^{1/2} \left[ \left(\int_0^T |\pi_t^{\lambda}| |\lambda_t'-\lambda_t|\, d\langle M\rangle_t\right)^2 \sup_{|\varepsilon| \le 1} \exp \left\lbrace  2 \varepsilon p \int_0^T |\pi_t^{\lambda}| |\lambda_t'-\lambda_t|\, d\langle M\rangle_t \right\rbrace  \right] . 
		\end{eqnarray*}
		Note that
		\begin{eqnarray*}
			\left|u^{\lambda'}-\mathbb{E}U\left(\mathcal{E}\left(\int_0^T \pi^{\lambda}_t\, dR^{\lambda'}_t\right)
\right)	\right| &\leq&\\
			|u^{\lambda'}-u^{\lambda}|+ \mathbb{E}\left|U\left(\mathcal{E}\left(\int_0^T \pi^{\lambda}_t\, dR^{\lambda}_t\right)\right)- U\left(\mathcal{E}\left(\int_0^T \pi^{\lambda}_t\, dR^{\lambda'}_t\right)\right)\right|.
		\end{eqnarray*}
		Here the first term is $\leq C_1\Vert\lambda-\lambda'\Vert_{\beta}$ with some
		constant $C_1>0$ by Theorem \ref{frechet}. If
		\begin{eqnarray*}
		C_2:=\mathbb{E} \exp\left\{2p\int_0^T \pi_t^{\lambda}\,dM_t+2p\int_0^T \pi_t^{\lambda} \lambda_t\, d\langle M\rangle_t
		-p\int_0^T (\pi_t^{\lambda})^2 \, d\langle M\rangle_t\right\}
		\end{eqnarray*} 
		is finite (which is a condition similar in spirit to those of \cite{larsen2014expansion}) 
		and $\exp \left\lbrace  2p \int_0^T |\pi_t^{\lambda}| |\lambda_t'-\lambda_t|\, d\langle M\rangle_t \right\rbrace < K$,
		for some constant $K$ then the second term is dominated by
		\begin{equation}\label{marvo}
		C_2 K \mathbb{E}^{1/2}\left(\int_0^T |\pi_t^{\lambda}| |\lambda_t'-\lambda_t|\, d\langle M\rangle_t\right)^2,
		\end{equation}
		and the latter expression is expected to be small if $\lambda-\lambda'$ is small
		in an appropriate sense, depending on the concrete class of models and on what
		we know about the optimizer $\pi^{\lambda}$. For instance, if $\pi^{\lambda}, \lambda', \lambda$ were
		bounded and \[
		l\, dt\leq d\langle M\rangle_t\leq L\, dt,
		\] 
		held for some constants $l,L>0$ then there would be such a $K$ and \eqref{marvo} 
		would be dominated by constant times
		$\Vert\lambda -\lambda'\Vert_{\beta}$, we omit further details.}
\end{re}

\section{Optimization with fBM}\label{sec_fBm}

We recall here the definition of fractional Brownian motion and some related properties. 
\begin{defi} \label{defi:fBm}
	The standard fractional Brownian motion $(B^H_t)_{t \in \mathbb{R}}$ with Hurst parameter $0 < H < 1$ is defined as a continuous centered Gaussian process with
	$$\text{Cov}(B^H_t, B^H_s) = \frac{1}{2}\left( |t|^{2H} + |s|^{2H} - | t - s|^{2H}\right), \qquad t, s \in \mathbb{R}.$$
\end{defi}
For $H> \frac{1}{2}$, the
increments of $B^H$ are positively correlated; for $H < \frac{1}{2}$, they are negatively correlated. Furthermore, for every $\kappa \in  (0,H)$, its sample paths are almost surely Hölder continuous with index $\kappa$. It is worth noting that for $H > \frac{1}{2}$, the process exhibits long-range dependence. In the following, we will often use the integral representation of $B^H$,
\begin{equation}\label{hox}
B^H_t = \int\limits_{\mathbb{R}} {K^H(t,s)dW_s}, \qquad H \in \left(0, \frac{1}{2} \right)  \cup \left( \frac{1}{2}, 1 \right), 
\end{equation}
where 
\begin{equation}\label{kernel_fBM}
K^H(t,s):= C(H) \left( (t-s)_+^{H-\frac{1}{2}} - (-s)_+^{H-\frac{1}{2}} \right),\ t\geq 0,\ s\in\mathbb{R},
\end{equation}
and $W$ is a Brownian motion. 
Here we adopt the usual convention $x_+ := \max(0, x)$. $C(H)$ is the normalizing constant and defined by $$C(H) :=\left(  \int\limits_0^{\infty} { \left( (1 + s)^{H - \frac{1}{2}} - s^{H - \frac{1}{2}} \right)^2ds} + \frac{1}{2H}\right)^{-\frac{1}{2}} = \frac{(2H \sin \pi H \Gamma(2H))^{\frac{1}{2}}}{\Gamma\left(H + \frac{1}{2}\right) }.$$ 
The following notations  will be used in the sequel to simplify our computations.
\begin{align*}
C_1(\alpha) &:= \int\limits_0^{\infty} {\left( (1+s)^{\alpha } - (s)^{\alpha} \right)^2ds} + \int_0^1{(1-s)^{2\alpha}ds},\\
C_2(\alpha) &:=  \int\limits_0^{\infty} {\left( (1+s)^{\alpha} \ln(1+s) - (s)^{\alpha} \ln s \right)^2ds} + \int_0^1{ \left(  (1-s)^{\alpha} \ln (1-s)\right)^2 ds}.
\end{align*}
If $\alpha \in \left( -\frac{1}{2}, \frac{1}{2} \right)$, the functions $C_1(\alpha)$ and $C_2(\alpha)$ are well-defined.

The use of fractional Brownian motion in stock price modeling has a long history. Most importantly, fBms may provide an explanation for the presence of the long-range dependence phenomenon in stock return, as suggested by many empirical studies, for example \cite{willinger1999stock}. In a seminal work, \cite{comte1998long} proposed to model log-volatility by an fBm with Hurst parameter $H > \frac{1}{2}$ in order to capture the long memory property. 
By contrast, \cite{fukasawa2011asymptotic} and \cite{gatheral2014volatility} suggest that the choice $H < \frac{1}{2}$ enables us to be   consistent with observed term structure of the volatility skew. 

In this paper, we discuss the effect of the Hurst parameter on the value function rather than the choice 
of this parameter. In Subsection \ref{model_1}, the market model is chosen such that only 
the drift exhibits long or negative memory. In Subsection \ref{model2}, the log-volatility takes memory into account.
\subsection{Model 1}\label{model_1}
We assume that the market price of risk evolves as a fractional Ornstein-Uhlenbeck process satisfying $d\lambda^H_t=-\alpha \lambda^H_t dt +  dB^H_t,$ where $\alpha>0$ and $B^H$ is a fractional Brownian motion with parameter $H \in (0,1)$. This equation admits an explicit solution, 
$$\lambda^H_t = e^{- \alpha t} \left(\lambda^H_0 + \int\limits_0^t {e^{\alpha u} dB^H_u}  \right).$$
Here, the stochastic integral with respect to fBM is simply a pathwise Riemann Stieltjes integral, see Proposition A.1 of \cite{cheridito2003fractional}. 
Assume that the price of a risky asset follows $S^H = \mathcal{E}(R^H)$, where
$$dR^H_t= \lambda^H_tdt + \rho dW_t + \sqrt{1-\rho^2} dB_t,$$
$\rho \in (0,1)$ and $B$ is a Brownian motion independent of $W$. Let $\mathbb{F}$ be the completed natural filtration
of $(W,B)$.

\begin{pro}\label{diff_H}
	The mapping
\begin{align*}
 (0,1) &\to \mathcal{D}_{\beta}\\
H &\mapsto (\lambda^H_t)_{t \in[0,T]}
\end{align*}
is Fr\'echet differentiable and its derivative is the process
\begin{equation}\label{d_lambda}
D\lambda^{H}_t := \int\limits_{\mathbb{R}} { \partial_H K^H(t,s) dW_s}  - \alpha e^{-\alpha t} \int\limits_0^t{\left( \int\limits_{\mathbb{R}} {\partial_H K^H(u,s)  dW_s}\right) e^{\alpha u}du}.
\end{equation}
\end{pro}
\begin{proof}
First, since $\lambda^H_t$ is a Gaussian random variable, we easily check that $\lambda^H \in  \mathcal{D}_{\beta}$. It suffices to prove that
\begin{equation}\label{diff}
\frac{1}{|\varepsilon|}\left\| \lambda^{H+\varepsilon} - \lambda^H - \varepsilon D\lambda^H \right\|_{\beta}\to 0 \qquad \text{as } \varepsilon \to 0.
\end{equation}
Using H\"older's inequality, we have
$$
\left\| \lambda^{H+\varepsilon} - \lambda^H - \varepsilon D\lambda^H \right\|_{\beta}^{2\beta} \le c \int\limits_0^T {\mathbb{E}\left[\left| \lambda^{H+\varepsilon}_t - \lambda^H_t - \varepsilon D\lambda^H_t \right|^{2\beta}\right]dt}. 
$$
In what follows, the constant $c$ may vary from line to line. We compute $\int\limits_0^t {e^{\alpha u} dB^H_u}  =   e^{\alpha t}B^H_t - e^{\alpha t}B^H_0 - \alpha\int\limits_0^t{B^H_ue^{\alpha u}du},$ by integrating by parts. Therefore, we obtain
\begin{align*}
\lambda^{H+\varepsilon}_t &- \lambda^H_t -  \varepsilon\partial_H\lambda^H_t = \underbrace {B^{H+\varepsilon}_t - B^{H}_t - \varepsilon \int\limits_{\mathbb{R}} { \partial_H K^H(t,s) dW_s}}_{I^1_{\varepsilon}} \nonumber \\
&- \underbrace {\alpha  \int\limits_0^t{e^{-\alpha (t-u)}\left( B^{H+\varepsilon}_u-B^{H}_u - \varepsilon \int\limits_{\mathbb{R}} {\partial_H K^H(u,s)  dW_s}\right) du}}_{ I^2_{\varepsilon}}.
\end{align*}
It is immediate that $$\varepsilon^{-2\beta}\left\| \lambda^{H+\varepsilon} - \lambda^H - \varepsilon D\lambda^H \right\|_{\beta}^{2\beta} \le c \varepsilon^{-2\beta} \int\limits_0^T {\left(  \mathbb{E}\left[ I^1_{\varepsilon}(t) ^{2\beta} \right]  + \mathbb{E}\left[I^2_{\varepsilon}(t)^{2\beta}\right] \right) dt}.$$

\noindent
\textbf{Estimation for $I^{\varepsilon}_1$.} By computing the $2\beta$-moment of the Gaussian random variable $I^1_{\varepsilon}$, we obtain
\begin{align*}
\varepsilon^{-2\beta}\mathbb{E}\left[ I^1_{\varepsilon}(t) ^{2\beta} \right] &\le c \left(\int\limits_{\mathbb{R}} { \left|\frac{K^{H + \varepsilon}(t,s) - K^{H}(t,s)}{\varepsilon} - \partial_H K^H(t,s) \right|^{2} ds} \right)^{\beta} 
\end{align*}
which tends to zero by Lemma \ref{lem:limit2}. Furthermore, if $|\varepsilon| < \delta$ for some $\delta > 0$, this term is bounded by $c \left( t^{2H + 1} ( 1 + (\ln t)^2 ) \right)^{\beta}$ if $t \ge 1$ and $c \left( t^{2H -2\delta} ( 1 + (\ln t)^2 ) \right)^{\beta}$ if $0\le t < 1.$

\noindent
\textbf{Estimation for $I^2_{\varepsilon}$.} On the one hand, we see that
\begin{align*}
&\varepsilon^{-2\beta}\mathbb{E}\left[ I^2_{\varepsilon}(t) ^{2\beta} \right]
\le \alpha   \int\limits_0^t{e^{- 2 \alpha \beta (t-u)} \mathbb{E}\left[\left| \frac{B^{H + \varepsilon}_u - B^H_u }{\varepsilon}  -  \int_{\mathbb{R}} { \partial_H K^H(u,s) dW_s}\right|^{2\beta}\right]du}  \\
&\le \alpha \left( \int\limits_0^t{e^{-2\alpha (t-u)} \left(\int\limits_{\mathbb{R}}  {\left| \frac{ K^{H + \varepsilon}(u,s) -  K^{H}(u,s)}{ \varepsilon} - \partial_H K^H(u,s)\right|^2 ds}\right) ^{\beta} du} \right)
\end{align*}
which tends to zero by Lemma \ref{lem:limit2} and the dominated convergence theorem. On the other hand, we obtain  
\begin{align*}
\varepsilon^{-2\beta}\mathbb{E}\left[ I^2_{\varepsilon}(t) ^{2\beta} \right] \le c \int_0^{t \wedge 1}{e^{-2\alpha(t-u)} \left( u^{2H -2\delta} ( 1 + (\ln u)^2 ) \right)^{\beta} du} \\
+ c \int_1^{t \vee 1} {e^{-2\alpha(t-u)} \left( u^{2H + 1} ( 1 + (\ln u)^2 ) \right)^{\beta} du} .
\end{align*}
Using the dominated convergence theorem again, the convergence (\ref{diff}) holds true. In conclusion, the Fr\'echet derivative of the process $\lambda^{H}$ is found.
\end{proof}

Putting together Theorem \ref{frechet} and Proposition \ref{diff_H}, we obtain the 
derivative of the value function $u^H(x):=u^{\lambda^H}(x)$ with respect to the Hurst parameter $H$.

\begin{theorem}
The mapping $H \mapsto u^H(1)$ from $(0,1)$ to $ \mathbb{R}$ is differentiable. Its derivative is given by
\begin{align*}
pu^{\lambda^{H}} \mathbb{E}\left[ \frac{d\tilde{\mathbb{P}}^{\lambda^{H}}}{d\mathbb{P}}\int\limits_0^T {(D\lambda^H_t)dR^{\lambda^{H}}_t} \right].
\end{align*}
\end{theorem}
The most interesting consequence of this result is perhaps the case $H = 1/2$. 
\begin{cor}\label{expan_1/2}
At $H = \frac{1}{2}$, we have 
$$
u^{\frac{1}{2}+ \varepsilon} - u^{\frac{1}{2}} = \varepsilon pu^{\lambda^{\frac{1}{2}}} \mathbb{E}\left[ \frac{d\tilde{\mathbb{P}}^{\lambda^{\frac{1}{2}}}}{d\mathbb{P}}\int\limits_0^T {(D\lambda^{\frac{1}{2}}_t)dR^{\lambda^{\frac{1}{2}}}_t} \right]  + o(\varepsilon),$$ 
and the first order term could be evaluated numerically.
\end{cor}
\begin{proof}
First, we compute the derivative of $\lambda^{H}$ at $H = \frac{1}{2}$,
\begin{align*}
D\lambda^{\frac{1}{2}}_t = C\left( \frac{1}{2}\right) \int_{-\infty}^0{ \ln(1 + t/s) dW_s} + \int_0^t{ \left( \partial_H C\left(\frac{1}{2} \right) + C\left(\frac{1}{2}\right)  \ln(t-s) \right) dW_s.} 
\end{align*}
Second, plugging this derivative into the G\^ateaux derivative of $u$ given in Theorem \ref{gateaux} gives us the derivative of $u$ at $H=1/2$:
\begin{equation}\label{der} 
pu^{\lambda^{\frac{1}{2}}} \mathbb{E}\left[ \frac{d\tilde{\mathbb{P}}^{\lambda^{\frac{1}{2}}}}{d\mathbb{P}}\int\limits_0^T {(D\lambda^{\frac{1}{2}}_t)dR^{\lambda^{\frac{1}{2}}}_t} \right]. 
\end{equation}
Recall that, by results of \cite{kim-omberg}, $u^{1/2}$ and $\tilde{\mathbb{P}}^{\lambda^{1/2}}$ 
are explicit modulo solving ordinary differential equations. Hence the first order term can indeed be calculated
with appropriate Monte Carlo simulations.
\end{proof}
The expansion in Corollary \ref{expan_1/2} gives us an approximation for the value function near 
$H = \frac{1}{2}$. As the optimization with $H \ne \frac{1}{2}$ is no longer Markovian and an efficient 
approach for finding its optimizer is unavailable, the above approximations are valuable. 


\subsection{Model 2} \label{model2}
In this subsection, we consider the case of memory in log volatility. The following model is suggested in \cite{gatheral2014volatility}
\begin{align*}
S^H_t &= \mathcal{E}(R^H)_t,\\
dR^H_t &= \mu(R^H_t)dt + \rho e^{\sigma^H_t}dW_t + \sqrt{1-\rho ^2}e^{\sigma^H_t}dB_t,\\
d\sigma^H_t &= -\alpha \sigma^H_t + dB^H_t.
\end{align*}
First, let us define an auxiliary market
\begin{align*}
\tilde{S}^H_t &= \mathcal{E}(\tilde{R}^H)_t,\\
d\tilde{R}^H_t &= \frac{\mu(R^H_t)}{e^{\sigma^H_t}}dt + \rho dW_t + \sqrt{1-\rho ^2}dB_t.
\end{align*}
This auxiliary market shifts the problematic factors involving long memory of the log volatility into the drift, which makes it possible to apply expansion results. To begin with, we will prove that this shifting argument preserve the structure of admissible strategies in the original market.
\begin{lemma}\label{auxiliary_strategy}
An strategy $\pi$ is admissible for the $S^H$-market if and only if $\tilde{\pi}:=\pi e^{\sigma^H}$ is an admissible strategy for the $\tilde{S}^H$-market. Furthermore, it holds that $X^{\pi} = X^{\tilde{\pi}}.$
\end{lemma}
\begin{proof}
Since $\pi$ is an admissible strategy in the $S^H$-market, one  has easily that $\int\limits_0^T {\pi^2_te^{2\sigma^H_t}dt} < +\infty.$
This is also the requirement for the admissibility of the strategy $\pi e^{\sigma^H}$ in the $\tilde{S}^H$-market. A straightforward computation shows that
$$dX^{\pi}_t = \pi_t dR^H_t = (\pi e^{\sigma^H_t})(e^{-\sigma^H_t}dR^H_t) = \tilde{\pi}_t d\tilde{R}^H_t = dX^{\tilde{\pi}}_t,$$
and this implies the second conclusion.
\end{proof}
As a result, Lemma \ref{auxiliary_strategy} tells us that the solution of the optimization problem in the $S^H$-market is also the solution of that problem in the $\tilde{S}^H$-market, which can be studied along the lines of Section \ref{model_1}. For simplicity, we assume $\mu:=1$.
\begin{pro}
The mapping $H \mapsto e^{-\sigma^H_t}$ is Fr\'echet differentiable and its derivative is the process $-e^{-\sigma^H_t}D\lambda^H_t$, where $D\lambda^H_t$ is exactly as in (\ref{d_lambda}). 
\end{pro}
\begin{proof}
We will prove that $ \varepsilon^{-1} || e^{-\sigma^{H + \varepsilon} }- e^{-\sigma^H} + \varepsilon e^{-\sigma^H} D\lambda^H ||_{\beta} \to 0$ as $\varepsilon$ tends to zero,
or equivalently,
$$ \mathbb{E}\left[ \left(  \int_0^T{\left| \frac{e^{-\sigma^{H + \varepsilon}_t }- e^{-\sigma^H_t}}{\varepsilon} +  e^{-\sigma^H_t} D \lambda^H_t\right| ^2dt}\right)^{\beta} \right]  \to 0.$$
Since the quantity $\varepsilon^{-1}(e^{-\sigma^{H + \varepsilon}_t } - e^{-\sigma^H_t} ) -  e^{-\sigma^H_t} D\lambda^H_t$ convergences to zero in $\mathbb{P} \times Leb$, it suffices to check its uniform integrability to establish the above convergence. The mean value theorem implies that
$$\frac{e^{-\sigma^{H + \varepsilon}_t + \sigma^H_t }- 1}{\varepsilon} =- e^{\tilde{\varepsilon}(-\sigma^{H + \varepsilon}_t + \sigma^H_t)} \frac{\sigma^{H + \varepsilon}_t - \sigma^H_t}{\varepsilon}$$
for some $|\tilde{\varepsilon}| < 1$. Therefore, for some $\beta' > \beta$, one estimates using H\"older's inequality,
\begin{align*}
\mathbb{E}\left[\left| \frac{e^{-\sigma^{H + \varepsilon}_t -\sigma^H_t }- 1}{\varepsilon} \right|^{2\beta'} \right] 
\le \mathbb{E} \left[ e^{4\beta'\tilde{\varepsilon}(-\sigma^{H + \varepsilon}_t + \sigma^H_t)} \right]^{1/2} \mathbb{E}\left[\left|   \frac{\sigma^{H + \varepsilon}_t - \sigma^H_t}{\varepsilon} \right|^{4\beta'} \right]^{1/2}.
\end{align*}
We notice that $\sigma^{H + \varepsilon}_t - \sigma^H_t$ is a Gaussian random variable with zero mean. An upper bound for its variance is found as follows
\begin{align*}
\mathbb{E}[(\sigma^{H + \varepsilon}_t - \sigma^H_t)^2] \le 2 \mathbb{E}\left[ (B^{H + \varepsilon}_t - B^{H}_t)^2 \right] + 2 \mathbb{E}\left[ \alpha^2\left( \int_0^t {e^{-\alpha(t-u)} (B^{H+\varepsilon}_u - B^H_u) du } \right)^2  \right] \\
\le 2 \mathbb{E}\left[ (B^{H + \varepsilon}_t - B^{H}_t)^2 \right]  +  2 \alpha^2 \int_0^t{e^{-2\alpha(t-u)}du}  \int_0^t{\mathbb{E}\left[ (B^{H+\varepsilon}_u - B^H_u)^2 \right] du }.
\end{align*}
Using Lemma \ref{lemma:compute_integral}, we obtain 
\begin{align*}
\mathbb{E}\left[ (B^{H + \varepsilon}_t - B^{H}_t)^2 \right] \le 2 \int_{\mathbb{R}} { |K^{H + \varepsilon}(t,s)|^2 + |K^{H}(t,s)|^2 ds}\le c(t^{2H + 2\varepsilon} + t^{2H}). 
\end{align*}
The variance is bounded by 
$$ c \left( t^{2H + 2 \varepsilon} + t^{2H} + \alpha(1-e^{-2\alpha t}) \left( \frac{t^{2H + 2\varepsilon}}{2H + 2\varepsilon + 1} + \frac{t^{2H+1}}{2H+1}  \right)  \right)  .$$
By using the moment generating function of a normal distribution, we find an upper bound for $\mathbb{E}[e^{4\beta' \tilde{\varepsilon} (-\sigma^{H + \varepsilon}_t +  \sigma^H_t)}]$ and hence, it holds that $$\sup_{|\varepsilon| < \delta} \int_0^T {\mathbb{E}\left[ e^{4\beta'\tilde{\varepsilon}(-\sigma^{H + \varepsilon}_t + \sigma^H_t)} \right] dt} < \infty.$$ One can prove, similarly to the argument in the proof of Proposition \ref{diff_H}, that $$\sup_{|\varepsilon| < \delta} \int_0^T {\mathbb{E} \left[\left|   \frac{\sigma^{H + \varepsilon}_t - \sigma^H_t}{\varepsilon} \right|^{4\beta'} \right]dt} < \infty.$$
 The uniform integrability is then satisfied, $$\sup_{|\varepsilon| < \delta} \mathbb{E}\left[ \int_0^T{\left| \frac{e^{-\sigma^{H + \varepsilon}_t }- e^{-\sigma^H_t}}{\varepsilon} \right|^{2\beta'}dt } \right] < \infty,$$
and thus the proof is complete.
\end{proof}

As usual, we define $u^{H}(x):= u^{e^{-\sigma^H}}(x)$ for the $\tilde{S}$-market. In Model 2, $u^{1/2}$ is not known explicitly but it seems possible to calculate it using the
corresponding Hamilton-Jacobi-Bellmann equation (recall that $H=1/2$ corresponds to a Markovian
optimal investment problem). We do not treat such ramifications in the present paper, neither do we
pursue the numerical evaluation of \eqref{der} above. These are left for future research.

\section{Further applications}\label{sec_further}

Fix $\mu\in\mathbb{R}$ and consider the process $\lambda^{\varepsilon}:=\mu+\nu^{\varepsilon}$ where
\begin{align*}
d\nu^{\varepsilon}_t = -\frac{1}{\varepsilon} \nu^{\varepsilon}_t dt + dW_t,\ \nu^{\varepsilon}_0=0,
\end{align*}
for $\varepsilon>0$, $W$ is standard Brownian motion and set $\nu^0_t:=0$,
$\lambda^0_t:=\mu$, $t\in [0,T]$.

When $M$ is a Brownian motion (possibly correlated with $W$), this choice of drift corresponds to 
continuous-time FADS models, see \cite{paolo}, where mean-reversion is intensifying to infinity as 
$\varepsilon\to 0$ (thus the limit is the constant drift $\mu$). 

Results of Section \ref{sec_frechet} enable us to get an asymptotic result on the value of
$u^{\varepsilon}(x):=u^{\lambda^{\varepsilon }}(x)$ when $\varepsilon\to 0$, i.e. when mean reversion is strong.

\begin{pro}\label{fd} Let $M$ be such that $d\langle M\rangle_t\leq G\,dt$ with some constant $G>0$. 
	Then, for each $0<\delta<1/2$, 
	\begin{equation}\label{palmira}
	u^{\varepsilon}-u^0=o(\varepsilon^{\delta}).
	\end{equation}
\end{pro}
\begin{re}{\rm 
	When $M=\rho W+\sqrt{1-\rho^2} B$ with a Brownian motion $B$ independent of $W$ and the 
	filtration is the one generated by $(W,B)$ then 
	the optimal utilities $u^{\varepsilon}$ can be explicitly calculated, see \cite{kim-omberg} and \cite{battauz2015kim}. 
	Proposition \ref{fd} derives 
	the asymptotics of the value functions for other choices of $M$ as well, where
	no explicit solution can be expected. We cannot determine the first order expansion
	but we can still estimate the magnitude of the deviation of the optimal utility $u^{\varepsilon}$ from
	that of the benchmark model, $u^0$.}
\end{re}

\begin{proof}[Proof of Proposition \ref{fd}.] Let us consider, instead of $\nu^{\varepsilon}$,
	$$
	d\nu^{\eta}_t = -\frac{1}{\eta^{\kappa}} \nu^{\eta}_t dt + dW_t,\ \nu^{\eta}_0=0,
	$$
	for some $\kappa>2$, setting $\lambda^{\eta}:=\nu^{\eta}+\mu$ and $\nu^0:=0$, 
	$\lambda^0:=\mu$.
	The explicit formula for $\nu^{\eta}$ is 
	\begin{equation}
	\nu^{\eta}_t = e^{- t/\eta^{\kappa} }\int\limits_0^t {e^{s/\eta^{\kappa}}dW_s}.
	\end{equation}
	We prove that the (Fr\'echet) right-hand derivative of $\eta\to \nu^{\eta}\in\mathcal{D}_{\beta}$ is $0$, i.e.
	\begin{align*}
	\left\| \frac{\nu^{\eta} - \nu^0}{\eta} \right\|^{2 \beta}_{\beta} &= 
	\mathbb{E}\left[ \left(\int\limits_0^T {\frac{(\nu^{\eta}_t)^2}{\eta^2} d\langle M\rangle_t} 
	\right)^{\beta} \right] &\leq 
G^{\beta}\mathbb{E}\left[ \left(\int\limits_0^T {\frac{(\nu^{\eta}_t)^2}{\eta^2} dt} 
	\right)^{\beta} \right] 
\to 0,\ \eta\to 0.
	\end{align*}
	In the sequel $C>0$ denotes a constant whose value may change from line to line.
	Denoting by $N^{\eta}$ the martingale $\int_0^{\cdot} {e^{s/\eta^{\kappa}}dW_s}$, we have
	\begin{align*}
	\mathbb{E}\left[ \left(\int\limits_0^T {(\nu^{\eta}_t)^2 dt} \right)^{\beta} \right] &= 
	\mathbb{E}\left[ \left(\int\limits_0^T 
	{e^{-2t/\eta^{\kappa}}( N^{\eta}_t)^2 dt} \right)^{\beta} \right]\\
	&\le \mathbb{E}\left[ \left(\int\limits_0^T {e^{-2t/\eta^{\kappa}}\sup_{s\in [0,t]}|N^{\eta}_s|^2 dt} \right)^{\beta} \right]\\
	&\le C\mathbb{E}\left[\int\limits_0^T {e^{-2\beta t/\eta^{\kappa}}\sup_{s\in [0,t]}|N^{\eta}_s|^{2\beta} dt}  \right]\\
	&\le C \int\limits_0^T {e^{-2\beta t/\eta^{\kappa}}\mathbb{E}\left[\sup_{s\in [0,t]}|N^{\eta}_s|^{2\beta} \right] dt}  .
	\end{align*}
	The Burkholder-Davis-Gundy inequality shows that
	$$\mathbb{E}\left[ \sup_{ s \in [0,t]}\left| N^{\eta}_s \right|^{2\beta} \right] \le C \left( \int\limits_0^t 
	{e^{2s/\eta^{\kappa}}ds}\right)^{\beta} = C \left( \frac{ e^{2t/\eta^{\kappa}} - 1}{ 2/\eta^{\kappa} } \right)^{\beta}$$
	hence
	\begin{align*}
	\mathbb{E}\left[ \left(\int\limits_0^T {(\nu^{\eta}_t)^2/\eta^2 dt} \right)^{\beta} \right]
	&\le \frac{C}{\eta^{2\beta}}  \int\limits_0^T {e^{-2\beta t/\eta^{\kappa}} \left( \frac{ e^{2t/\eta^{\kappa}} - 1}{ 2/\eta^{\kappa}} \right)^{\beta}dt} \to 0,
	\end{align*}
	as $\eta$ tends to zero. This clearly entails that the derivative of $\eta\to \lambda^{\eta}$ is also
	$0$.
	Noting the Fr\'echet differentiability of $\lambda\to u^{\lambda}$ we get that
	$u^{\eta}-u^0=o(\eta)$ where $u^{\eta}$ is the optimal utility for the drift process $\nu^{\eta}$. Now apply this with 
	$\kappa:=1/\delta$ and $\varepsilon:=\eta^{\kappa}$ to conclude that
	\eqref{palmira} indeed holds true.
\end{proof}

\section{Conclusions}\label{conci}

Here we discuss certain important ramifications that deserve to be addressed
in more detail. A relevant question is whether one may get an
approximation of the optimal strategy $\pi^{\lambda}$ instead of the value function 
$u^{\lambda}$. As we have already pointed out in Remark \ref{mascot} above,
Subsection 3.2 of \cite{larsen2014expansion} mentions such a result.

The arguments of \cite{larsen2014expansion} work along a specific directional
derivative, i.e. they study $\pi^{\lambda+\varepsilon\lambda'}$ with fixed $\lambda'$ and require technical assumptions that are
not easy to check (they are in terms of $\tilde{\mathbb{P}}^{\lambda}$ about which we 
have little information in general), see Theorem 3.4 of \cite{larsen2014expansion}. 
A natural continuation of our present work would be to strengthen
their arguments so as to yield the ``almost optimality'' of \eqref{lal} uniformly
along all directions $\lambda'$ and preferably under sufficient conditions
that are easy to check.

We shortly explain in a one-period model of a financial market what such an extension
should involve. Let us
try to maximize 
\[
\mathbb{E}[U(\phi R^{\lambda})],
\] 
where, for convenience, we assume that the concave and continuously differentiable
utility function $U$ is defined and finite
on $\mathbb{R}$, $\phi\in\mathbb{R}$ is 
the number of units of the stock that is bought (this represents the portfolio strategy), and
the random variable $R^{\lambda}$ is the return on the given stock, parametrized by some 
$\lambda\in\mathbb{R}$. (We assume that the agent has $0$ initial capital.)

Under appropriate conditions, the optimal strategy $\phi(\lambda)$ satisfies
\[
\mathbb{E}[U'(\phi(\lambda) R^{\lambda})R^{\lambda}]=0.
\]
If the implicit function theorem is applicable then $\lambda\to \phi(\lambda)$
is continuously differentiable and its derivative satisfies
\begin{eqnarray*}
	\partial_{\lambda}\phi(\lambda)\mathbb{E}[U''(\phi(\lambda)R^{\lambda})(R^{\lambda})^2]
	&+& \mathbb{E}[U''(\phi(\lambda) R^{\lambda})\phi(\lambda)R^{\lambda}\partial_{\lambda}R^{\lambda}]\\
	&+&
	\mathbb{E}[U'(\phi(\lambda) R^{\lambda})\partial_{\lambda}R^{\lambda}]=0.
\end{eqnarray*}
From this we can express $\partial_{\lambda}\phi(\lambda)$ and obtain the first-order correction term of $\lambda\to\phi(\lambda)$ for $\lambda$ 
in a neighbourhood
of some fixed $\lambda^*$, i.e. $\phi(\lambda)\approx \phi(\lambda^*)+(\lambda-\lambda^*)
\partial_{\lambda}\phi(\lambda^*)$. Higher order correction terms can be deduced
in an analogous way.

It is quite clear that transferring the above arguments to the continuous-time
setting looks highly non-trivial: what kind of infinite-dimensional Banach spaces should be used for the strategy and for the (drift) parameter so that a similar argument
can go through? From Subsection 3.2 of \cite{larsen2014expansion} we
do know the first-order correction term (though it is not explicit 
enough due to the presence of $\gamma^B$, see \eqref{lal} above) hence it would remain
to establish that the obtained representations hold uniformly in all directions.

An extension of the optimal strategy in the model space thus requires further research effort which has started but certainly leads beyond the scope of the present paper.

\section{Appendix}\label{sec_appen}

In this section, we provide some useful results for fractional Brownian motion.
\begin{lemma}\label{lemma:compute_integral}
The functions $C_1(\alpha), C_2(\alpha)$ are continuous on $\left( -\frac{1}{2}, \frac{1}{2} \right)$. In addition, we have that
\begin{align*}
\int\limits_{\mathbb{R}} { \left( (t-s)_+^{\alpha} - (-s)_+^{\alpha } \right)^2ds} &= t^{2\alpha + 1}C_1(\alpha),\\
\int_{\mathbb{R}}{\left( (t-s)_+^{\alpha} \ln (t-s)_+ - (-s)_+^{\alpha} \ln (-s)_+ \right)^2ds} &\leq 2 t^{2\alpha + 1} (\ln t)^2 C_1(\alpha) +  2 t^{2\alpha + 1} C_2(\alpha).
\end{align*} 
\end{lemma}
\begin{proof}
The continuity of $C_1$ comes from the continuity of $C$. We now consider the function $C_2$. Letting $\varepsilon \to 0$, it holds that 
$$|(1+s)^{\alpha + \varepsilon}\ln(1+s) - s^{\alpha + \varepsilon}\ln s| \to |(1+s)^{\alpha}\ln(1+s) - s^{\alpha}\ln s|.$$
Fix $\delta < \alpha$. We separate two cases.  \\
\textbf{The case $\alpha \in (0, 1/2)$.}  For $|\varepsilon| < \delta$ then
$$|(1+s)^{\alpha + \varepsilon}\ln (1+s) - s^{\alpha + \varepsilon} \ln s| \le |(1+s)^{\alpha + \delta} \ln (1+s) - s^{\alpha + \delta} \ln s|, \text{  if $s > 1$}.$$ 
If $0 < s \le 1$, since $s^{\alpha + \varepsilon}\ln s \to 0$ as $s$ tends to zero, we can find a constant such that $|(1+s)^{\alpha + \varepsilon}\ln (1+s) - s^{\alpha + \varepsilon} \ln s| \le c$. Thus, the dominated convergence theorem implies $C_2(\alpha + \varepsilon) \to C_2(\alpha)$.\\
\textbf{The case $\alpha \in (-1/2,0).$} Let $|\varepsilon| < \delta$, there exists $K$ big enough such that
$$|(1+s)^{\alpha + \varepsilon}\ln (1+s) - s^{\alpha + \varepsilon} \ln s| \le |(1+s)^{\alpha - \delta} \ln (1+s) - s^{\alpha - \delta} \ln s|, \text{ if $s > K$.}$$
If $s < K$, we estimate $|(1+s)^{\alpha + \varepsilon}\ln (1+s) - s^{\alpha + \varepsilon} \ln s| \le m + s^{\alpha + \varepsilon} \ln s$
for some $m>0$. Furthermore, we compute
\begin{align*}
\int_0^K{(s^{\alpha+\varepsilon} \ln s)^2}ds &= \left. \frac{s^{2(\alpha + \varepsilon) + 1} (\ln s)^2}{2(\alpha + \varepsilon) + 1}\right|_0^K \\
&- \left. \frac{1}{2(\alpha + \varepsilon) + 1} s^{2(\alpha + \varepsilon) + 1}\left( \frac{\ln s}{2(\alpha + \varepsilon) + 1} - \frac{1}{(2(\alpha + \varepsilon) + 1)^2}\right)\right|_0^K.
\end{align*}
Now, applying the extended dominated convergence theorem to the sequence $|(1+s)^{\alpha + \varepsilon}\ln (1+s) - s^{\alpha + \varepsilon} \ln s|^2$ which is smaller than
$$1_{s > K}|(1+s)^{\alpha - \delta} \ln (1+s) - s^{\alpha - \delta} \ln s|^2 + 1_{s \le K} (m^2 + s^{2(\alpha + \varepsilon)} \ln^2 s),$$ we obtain $C_2(\alpha + \varepsilon) \to C_2(\alpha)$.

Let us consider the first equality. By changing of variable $s = tx$, one has
\begin{align*}
\int\limits_{\mathbb{R}} { \left( (t-s)_+^{\alpha} - (-s)_+^{\alpha} \right)^2ds} &= \int\limits_0^{\infty} {\left( (t+s)^{\alpha} - s^{\alpha} \right)^2ds} + \int\limits_0^t {(t-s)^{2\alpha}ds}\\
&= t^{2\alpha + 1} \left( \int\limits_0^{\infty} {\left( (1+x)^{\alpha } - (x)^{\alpha} \right)^2ds} + \int_0^1{(1-x)^{2\alpha}dx} \right)
\end{align*}
and the first identity follows. Now, straightforward computations show that 
$$
(t+s)^{\alpha} \ln(t+s) - s^{\alpha}\ln s = t^{\alpha} \ln t \left[ (1 + x)^{\alpha} - x^{\alpha}  \right] + t^{\alpha}\left[ (1+x)^{\alpha}\ln(1+x) - x^{\alpha} \ln x \right]  
$$
and $(t -s)^{\alpha}\ln(t-s) = t^{\alpha}\ln t (1-x)^{\alpha} + t^{\alpha}(1-x)^{\alpha} \ln(1-x)$. 
Therefore, 
\begin{align*}
&\int_{\mathbb{R}}{\left( (t-s)_+^{\alpha} \ln (t-s)_+ - (-s)_+^{\alpha} \ln (-s)_+ \right)^2ds} = \int_0^t {\left( (t-s)^{\alpha} \ln (t -s)\right) ^2ds} \\
&+\int_0^{\infty}{\left( (t+s)^{\alpha} \ln (t+s) - s^{\alpha} \ln s \right)^2ds} \le 2t^{2\alpha + 1} (\ln t)^2 \int_0^1{(1-x)^{2\alpha}dx} \\
&+ 2t^{2\alpha +1} \int_0^1{\left[ (1-x)^{\alpha}\ln(1-x)\right] ^2dx} + 2 t^{2\alpha + 1}(\ln t)^2 \int_0^{\infty}{\left((1+x)^{\alpha} - x^{\alpha} \right)^2dx }\\
& + 2t^{2\alpha + 1}\int_0^{\infty}{\left( (1+x)^{\alpha} \ln(1+x) - x^{\alpha} \ln x \right)^2dx }.
\end{align*}
From that, the proof is complete.
\end{proof}

\begin{lemma}\label{lemma:diff_kernel} The derivative of the function $K^H$ is

\begin{align}\label{diff_kernel_fBM}
\partial_H K^H(t,s) &= \partial_H C(H) \left( (t-s)_+^{H-\frac{1}{2}} - (-s)_+^{H-\frac{1}{2}} \right) \nonumber \\
&+ C(H) \left( (t-s)_+^{H-\frac{1}{2}} \ln (t-s)_+ - (-s)_+^{H-\frac{1}{2}} \ln (-s)_+ \right).
\end{align}
Furthermore, the derivative is square integrable over $\mathbb{R}$.
\end{lemma}
\begin{proof}
Since the function $C(H)$ is smooth in $H$, we get the derivative of $K^H$ with respect to $H$. The second statement is deduced from Lemma \ref{lemma:compute_integral}.
\end{proof}
\begin{lemma}\label{lem:limit2}
As $\varepsilon$ goes to zero, the following convergence holds, for any $t\geq 0$,
\begin{equation}\label{diff_kernel_converge}
\int\limits_{\mathbb{R}} { \left| \frac{K^{H + \varepsilon}(t,s) - K^{H}(t,s)}{\varepsilon} -  \partial_H K^H(t,s) \right|^{2} ds} \to 0.
\end{equation}
In addition, the integral is bounded by 
$$\begin{cases}
c t^{2H + 1} ( 1 + (\ln t)^2 ) &\text{if}\ t\ge 1\\
c t^{2H -2\delta} ( 1 + (\ln t)^2 ) &\text{if}\ 0\le t < 1, |\varepsilon| < \delta < H,
\end{cases} $$
which is a function in $L^{\beta}(0,T)$ for $\beta > 1.$
\end{lemma}
\begin{proof}
First, we notice that 
$$ \lim_{\varepsilon \to 0} \left[  \frac{K^{H + \varepsilon}(t,s) - K^{H}(t,s)}{\varepsilon} -  \partial_H K^H(t,s) \right]  = 0, \qquad \text{for each $t,s$}.$$
By the mean value theorem, $ K^{H + \varepsilon}(t,s) - K^{H}(t,s) = \varepsilon \partial_H K^{H + \tilde{\varepsilon}}(t,s)$ for some $\tilde{\varepsilon} = \tilde{\varepsilon}(t,s).$ Thus, it is immediate that 
\begin{align*}
\left| \frac{K^{H + \varepsilon}(t,s) - K^{H}(t,s)}{\varepsilon} -  \partial_H K^H(t,s)\right| ^2 \le 2(|\partial_H K^{H + \tilde{\varepsilon}}(t,s)|^2 + |\partial_H K^{H }(t,s)|^2).
\end{align*}
Since $\partial_H K^{H + \tilde{\varepsilon}}(t,s) \to \partial_H K^{H}(t,s)$ for all $t, s$ when $\varepsilon \to 0$ and
\begin{equation*}
\lim_{\varepsilon \to 0} \int_{\mathbb{R}} {2(|\partial_H K^{H + \tilde{\varepsilon}}(t,s)|^2 + |\partial_H K^{H }(t,s)|^2) ds} = 4 \int_{\mathbb{R}} {|\partial_H K^{H }(t,s)|^2 ds} 
\end{equation*}
by the formula (\ref{diff_kernel_fBM}) and Lemma \ref{lemma:compute_integral}, an application of the extended dominated convergence theorem  implies that (\ref{diff_kernel_converge}) holds true.  
We also estimate 
\begin{align*}
&\int_{\mathbb{R}}{ |\partial_{H}K^{H + \tilde{\varepsilon}}(t,s)|^2 + |\partial_{H}K^{H}(t,s)|^2 ds} \le 2 (\partial_H C(H + \tilde{\varepsilon}))^2 C_1(H - 1/2 + \tilde{\varepsilon}) t^{2H + 2 \tilde{\varepsilon}} \\
&+ 4C^2(H+\tilde{\varepsilon}) \left[ C_1(H - 1/2 + \tilde{\varepsilon}) t^{2H + 2 \tilde{\varepsilon}}(\ln t)^2 + C_2(H - 1/2 + \tilde{\varepsilon}) t^{2H + 2 \tilde{\varepsilon}}\right] \\
&+  2 (\partial_H C(H))^2 C_1(H - 1/2) t^{2H} 
+ 4C^2(H) \left[ C_1(H - 1/2) t^{2H}(\ln t)^2 + C_2(H - 1/2 ) t^{2H}\right].
\end{align*}
Since the functions $C(\cdot),\partial_H C(\cdot), C_1(\cdot), C_2(\cdot)$ are continuous with respect to $\alpha$, there exists $\delta > 0$ small enough, such that for all $\varepsilon \in (-\delta, \delta)$, one has
$$
\int\limits_{\mathbb{R}} { \left| \frac{K^{H + \varepsilon}(t,s) - K^{H}(t,s)}{\varepsilon} -  \partial_H K^H(t,s) \right|^{2} ds} \le 
\begin{cases}
c t^{2H + 1} ( 1 + (\ln t)^2 ) &\text{if}\ t\ge 1\\
c t^{2H -2\delta} ( 1 + (\ln t)^2 ) &\text{if}\ 0\le t < 1.
\end{cases}
$$
with some constant $c = c(H, \delta) >0$.
\end{proof}

\begin{theorem}[The extended dominated convergence theorem]\label{thm:extended_DCT}
Let $(\Omega, \mathcal{F}, \mu)$ be a measure space and let $f_n, g_n : \Omega \to \mathbb{R}$ be measurable functions such that $|f_n| \le g_n, \text{ } a.e.$ for all $n \ge 1$.  Suppose that
\begin{itemize}
\item $g_n \to g, a.e$ and $f_n \to f, a.e.$
\item $g_n, g \in L^1(\Omega)$ and $\int {|g_n| d\mu} \to \int {|g| d\mu}$ as $n \to \infty.$
\end{itemize}
Then $f \in L^1(\Omega)$,
$$ \lim_{n \to \infty} \int{f_nd\mu} = \int{f d\mu}, \qquad \text{and} \qquad \lim_{n \to \infty} \int{|f_n - f|d\mu} = 0.$$
\end{theorem}
\begin{proof}
See Theorem 2.3.11 of \cite{athreya2006measure}.
\end{proof}
\bibliography{memory}
\end{document}